\documentclass[10pt, conference, compsocconf]{IEEEtran}
\usepackage{multirow}
\usepackage{graphicx}

\usepackage{color}
\usepackage{algorithm}
\usepackage{algorithmic}
\usepackage{subfigure}
\usepackage{verbatim}
\usepackage{array}
\usepackage{amsmath}
\usepackage{amssymb}
\usepackage{amsfonts}
\usepackage{mathrsfs}
\usepackage{url}

\newtheorem{theorem}{Theorem}
\newtheorem{lemma}{Lemma}

\begin{document}

\title{An approach towards debiasing user ratings}

\author{\IEEEauthorblockN{Abhinav Mishra}
\IEEEauthorblockA{School of Computer Science\\
Georgia Institute of Technology\\
Atlanta, Georgia 30332--0250\\
Email: amishra41@gatech.edu}
}

\maketitle

\begin{abstract}
	With increasing importance of e-commerce, many  websites 
	have emerged where users can express their opinions about products, such as
	movies, books, songs, etc.  Such interactions can be modeled as
	bipartite graphs where the weight of the directed edge from a
	user to a product denotes a rating that the user imparts to the
	product.  These graphs are used for recommendation systems and discovering
	most reliable (trusted) products.  For these applications, it is important
	to capture the bias of a user when she is rating a product.  Users have
	inherent bias---many users always impart high ratings while many others
	always rate poorly.  It is necessary to know the bias of a reviewer while
	reading the review of a product.  It is equally important to compensate for
	this bias while assigning a ranking for an object.  In this paper, we
	propose an algorithm to capture the bias of a user and then subdue it to
	compute the true rating a product deserves.  Experiments
	show the efficiency and effectiveness of our system in capturing the bias of
	users and then computing the true ratings of a product.
\end{abstract}

\begin{figure}[!h]
\centering
\includegraphics[scale=0.64]{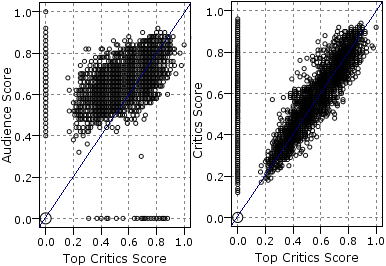}
\label{fig:error}
\caption{Audience in general have propensity to give high scores with respect to the top critics score.
 However, average score of critics appears to be more correlated with the score of top critics.}
\end{figure}  

\section{Introduction}
\label{sec:intro}

With the growing popularity of the Internet and e-commerce, many new websites
 have emerged where users can express their opinions about products,
such as movies, books, songs, etc.  Amazon, for example,
stores millions of products that users can search, browse, buy and rate.
Consequently, it records millions of hits per day and is one of the most
important product websites that has ever existed.  Similarly, IMDB
 lets users get informed about movies and also allows them
to read other users' reviews as well as rate movies themselves.

Most such product rating websites can be  modeled as (directed)  graphs
where users and products are modeled as nodes and the rating of a user for a
product is captured by adding a directed edge from the user node to the product
node.  The graph is bipartite as only users rate products.  The weight on
the edge captures the ``rating'' that the user provides to the product.  The
ratings can be used to rank products as well as recommend them to other users.

Users have inherent bias while rating a product---many users always impart high
ratings while many others always rate poorly.  For applications in
recommendation-based systems where products are ranked according to their
ratings, it is important to capture the bias of a user when she is rating a
product. In Figure 1, we show the average ratings given to the movies by
audience and critics. These ratings are collected from
{\url{www.rottentomatoes.com}} and are available as a public
dataset\footnote{\url{http://www.grouplens.org/node/462}}.  Values on the axes
are mostly missing values in the dataset.  The website
\url{www.rottentomatoes.com} also scores the movies by aggregating the reviews
from top critics; these scores are deemed more trustworthy. We can observe that
audience score is overwhelmingly positive with respect to the top critics score,
where as the average score of all the critics is highly correlated with those of
top critics.  This phenomenon is well observed \cite{chenkdd11} and is
attributed to a difference between quality and user rating.  It implies that
user ratings have inherent bias and are not trustworthy. 

While global ratings (for example, computed as an average of all ratings) are
less prone to individual biases of users, individual ratings are susceptible.
Many users read the ratings and reviews of other users when
trying to decide whether to buy a product or not.  Thus, it is necessary to
know the bias of a reviewer while reading the review of a product.  It is
equally important to compensate for this bias while assigning a ranking to an
object.  It is also observed that for the  objects that receive very few ratings,  a slight bias
on the part of users can lead to a significant change in their ratings \cite{chenkdd11}. In experiments,
we observe such phenomenon as well.

Product ratings usually follow the J-shaped distribution~\cite{bias09}, i.e.,
most of the ratings are overwhelmingly positive along with a small but
significant number of high negatives. In Figure 1, we  infer similar behavior as well. These ratings do not follow the expected
normal distribution.  This also indicates that the average rating may not be a
right metric to judge product quality because of
bias~\cite{ec06,bias09}. In this paper we propose a technique which captures
user bias and factors it while computing the rating of an object. We follow the intuition that if a user is known to 
exaggerate her ratings, i.e., user is positively biased, then we  compensate for her exaggeration by reducing the ratings.
Likewise if a user is known to give low scores, then we increase the ratings. 

Specifically,
our main contributions are as follows:
\begin{enumerate}
\item We propose an iterative technique, where we model both bias and true rating of an item and
 derive a mutually-recursive formulation. 
\item We show that our technique converges to a unique solution, i.e., it 
does not depend on the initial conditions. Further, we prove that error after any iteration is bounded, and
thus, maximum number of iterations required can be fixed apriori.

\item We conduct an in-depth evaluation with the earlier work, current schemes used in the modern system, 
and more specifically we show that our technique adheres well with the ground truth  in terms of absolute error, error in ranking, etc.
\item We also analyze the behavior and error associated with items that receive fewer ratings. We conjectured
that such items are likely to be impacted more by user bias. We experimentally
verify  and analyze this phenomenon in detail.
\end{enumerate}

In section \ref{sec:rwork}, we discuss causes of user bias and related approaches for debiasing.
In section \ref{sec:algo}, we formally define bias of a user and true rating associated with items. We present the algorithm, its 
interpretation, iterative technique to compute it, and finally we discuss the complexity of the algorithm. 
We show the necessary properties such as convergence, uniqueness, and the error-bound in Section \ref{sec:prop}.
 In Section \ref{sec:exp},
we give the experimental results such as performance of our algorithm and other approaches (Section \ref{sec:eval}),
effect of sparse ratings, analysis of ranking, convergence analysis (Section \ref{sec:qual}). Finally, we conclude in Section \ref{sec:conc}.

\section{Related Work}
\label{sec:rwork}
In this section, we review prior work about the origin of bias and argue that users are indeed biased.
We then give a overview of related approaches to compute bias.

\subsection{Sources of Bias}

The
bias we usually refer to is known as \emph{cognitive bias}.  It is the
human inability to judge the situation based on evidences alone and is influenced by cognitive factors as well.
It has been studied in great detail in the area of
cognitive science and social psychology. This was introduced by Tversky and Kahneman \cite{Kahneman_72}.
In an experiment \cite{tversky-kahneman-74}, they presented a numerical estimation of $8!$ (shown in a different manner) to two groups.
Two groups came up with a significantly different estimates.
When we rely too heavily on few parameters,
the bias is known as \emph{anchoring bias}.

\subsection{Bias in product ratings}
In this subsection, we focus on biases that affects product ratings.
 First we consider
\emph{affect heuristic}\footnote{It is a heuristic in which current affect influences decisions.}.
In general, if a reviewer is presented with large positive ratings and reviews,
then she is more likely
to give positive scores \cite{Carlson2010}. Also sorting the reviews based on
ratings 
triggers the affect heuristic as well. \emph{Purchasing bias} \cite{bias09} is another example, where people who value a product
highly are likely to buy and will not usually write a poor review.
Even if we assume that an online rating system (for a user to rate)  has been built perfectly including capturing the minute details of an object such as
packaging of the object, color, etc., a truthful user may still not conceive it perfectly. For example, it has been observed that user is likely to give high
score on Amazon than on eBay \cite{wolf2011}. This is because human judgments are subject to a number of bias, including focusing more on few parameters (Amazon emphasizes strength parameter \cite{wolf2011}) of the evidence  \cite{Kahneman_72}.


\subsection{Existing Approaches}

In this section, we review related algorithms proposed to compute bias.  In all
the algorithms, the main goal was to identify the \emph{net bias} rather than
computing different biases for each user. 
The first  such work \cite{kdd06} noticed that product evaluation depends
greatly on the evaluator. They proposed an algorithm to compute bias and showed
that  it outperforms standard statistical measures. However, their technique looks locally to 
compute true rating (\cite{kdd06} refer it as consensus), i.e., when source bias is not 
considered. This will work if these biased
ratings can give an unbiased estimate of true rating. However, this is not true.
Ranking algorithms like HITS \cite{hits} and PageRank \cite{pagerank}
look at the global structure, e.g., opinion of a better node should be given
more weight. In our case, opinion of a highly biased user should not be
directly used to calculate true rating. Recently, \cite{mcgl} proposed
many heuristics to compute true ratings based on local approaches.
Later, \cite{sdm07} authors improved their existing system by considering the global
structure as well. 
In experiments, we consider this algorithm~\cite{sdm07} as the  benchmark and give in-depth comparison with our technique.


\section{Algorithm}
\label{sec:algo}

The system is modeled as a directed bipartite graph $G = (V = U \cup O, E
\subseteq U \times O)$ where $U$ and $O$ denote the set of users and objects
(products) respectively.  The only edges allowed are those from $U$ to $O$,
i.e., users rate objects.  Thus, users have only outgoing edges and objects
have only incoming edges.  We constraint the user to give at most one rating to
a product.

The weight of a directed edge, denoted by $w_{ij}$, captures the rating given by
user $u_i \in U$ to object $o_j \in O$.  It is assumed that if the weight is
high, the user is rating the object positively; on the other hand, if it is low,
it is a negative opinion.  In this work, we assume the weights
to be between $0$ and $1$, i.e., $w_{ij} \in [0,1]$.  The \emph{neutral} rating
is somewhere in between, but not necessarily at $0.5$.

The number of edges in such a graph
may be large,
even if the number of users and objects are not very large.
In most real life graph datasets, however, users rate only a few
products, and hence, the size of the graph remains within a small factor of the
number of nodes.

The standard method to compute the net rating of any item $i$ is to consider the
average of all the ratings item $i$ receives ($avg_j\{w_{ji}\}$).  The method
may have some filters such as frequent users, but in general, the mean is
considered to be the best estimate. In our work, we aim to remove bias from each
individual rating, and then we take the mean. If $f(w_{ji})$ is an estimate of
the bias-free rating of a user to a product, then according to our formulation,
$avg_j\{f(w_{ji})\}$ is the net bias-free rating of the item $i$.

\subsection{Model}
\label{sec:def}

In this section, we precisely define the terms bias of a user $u_i$, denoted by
$b_i$, and true rating of an object $o_j$, denoted by $r_j$.
Further, $d^+_i$ denotes the set of objects a user $u_i$ rates and
$d^-_j$ denotes the set of users that an object $o_j$ receives ratings from.
Then,
\begin{align}
	\sum_{u_i \in U}\left|d^+_i \right|= \sum_{o_j \in O}\left|d^-_j \right|.
\end{align}

If we know that a user is gives higher ratings than usual, then we expect 
her ratings to be less trustworthy. In fact, we try to factor in such exaggeration in
ratings while coming up with a more realistic rating. We refer to such exaggeration (either high or low)
 in ratings as \emph{bias}, and the realistic rating as \emph{true rating}.  In our approach, we compute 
bias of each user assuming that we know the true ratings of all object.  Likewise, we compute true rating of an object by assuming that we know the 
bias of all object. We create an simple iterative system of two mutually
recursive variables,
bias and true rating.
The precise details are presented next.

\subsubsection{Bias}

The \emph{bias} of a user is defined as the average deviation of ratings given
by her from the true ratings:
\begin{align}
	\label{eq:bias}
	b_i = \frac{1}{\left|d^+_i \right |} \sum_{o_j \in d^+_i} {\big( w_{ij} - r_j \big)}.
\end{align}
In the above equation, $w_{ij} - r_j$ measures the deviation of the user $i$ while rating 
an object $j$. Here $w_{ij}$ denotes the rating given by the user $i$ to the object $j$, and $r_j$
is the true rating of the object $j$.
 For now, we assume that the true ratings are known. 
Later, we show that the algorithm converges to a unique solution, irrespective
of the initial values of true ratings assumed.

If a user gives higher ratings to objects than the true ratings those objects
deserve, the bias of the user will be positive, and the user is said to be
positively biased. Such intuition is clearly captured in Eq.~(\ref{eq:bias}).
Bias of a user can thus be seen as expected exaggeration in her rating an object.
Assuming the true ratings also lie between $[0,1]$, the
bias values range from $-1$ to $1$, i.e., $b_i \in [-1,1]$. 
If $b_i<0$, it means that the user $i$ usually gives lower scores, and likewise
if $b_i>0$ indicates that the user give higher ratings.

It can be argued that  products also have biases, i.e., some
products are likely to get high ratings and some are expected to get poor
ratings.  For example, a good movie is expected to attract high ratings from the
viewers.  If $\delta$ be the average rating of all the products and $b_o$ be the
bias of object $o$, then clearly, true rating of the object $o$ is $\delta+b_o$.
Therefore, the true rating implicitly captures the product bias, and thus we do not have to
model item bias separately.

\subsubsection{True Rating}

If we know that the bias of a user, i.e., she gives high/low scores than the average, then
we can adjust  her rating by factoring in bias. For example, if a user has a positive bias, then her ratings are more than
 what they should be.
Hence, to offset the effect of bias, some fraction of the  ratings
should be reduced.  Similarly, if a user is negatively biased, some fraction
of the ratings should be increased. Intuitively, the transformed weight would be the user's rating minus the bias, i.e., $(w_{ij} -  b_i)$.
But this formulation does not guarantee a convergence. Therefore, we make a relaxation to the formulation.


A simple
relaxation for a new rating, denoted by $w'_{ij}$,  after factoring bias can be written as:
\begin{align}
	w'_{ij} = w_{ij} - \alpha \cdot b_i.
\end{align}
Here we increase the rating, if bias is negative, and decrease if bias is
positive.
We restrict $\alpha$ to $(0,1)$. For technical reasons, such 
as convergence, we are forced to keep $\alpha<1$. This will be made clear in later sections.  
If $\alpha$ is high, we remove more bias and likewise, if $\alpha$ is low, we change the weight by a smaller amount.
At the same time, keeping a parameter is also helpful in many situations. For example, if we know that there are few trusted users such as
top critics, who are considered as unbiased users. In such a scenario, we can keep $\alpha=0$ for such users. Similarly, if we know that
there are few users who are unreliable, then we can keep a high $\alpha$ for such users.

In the above formulation, it should be noted that the new weight can be outside
the range $[0,1]$.  To overcome this undesirable property, we define a hard
cut-off at the boundaries, i.e., $0$ and $1$.
\begin{align}
	\label{eq:w}
	w'_{ij} =
	\left\{
		\begin{array}{cl}
			0 & \text{ if } w_{ij} - \alpha \cdot  b_i \leq 0 \\
			1 & \text{ if } w_{ij} - \alpha \cdot b_i \geq 1 \\
			w_{ij} - \alpha \cdot b_i & \text{ otherwise }
		\end{array}
	\right.
\end{align}
This ensures that whenever the new weight is more than $1$, it is
reset to $1$ and similarly, whenever it is negative, it is reset to
$0$.
Eq.~\eqref{eq:w} can be succinctly written as
\begin{align}
	\label{eq:tw}
	w'_{ij} = \max\big\{0, \min\{1, w_{ij} - \alpha \cdot b_i\}\big\}.
\end{align}
The above function to compute $w'_{ij}$ is our de-biasing function (mentioned as
$f(\cdot)$ earlier), and
it returns the unbiased ratings.
Using unbiased ratings, the \emph{true rating} of an object is defined as the
average of the unbiased ratings:
\begin{align}
	\label{eq:rating}
	r_j = \frac{1}{|d^-_j|} \sum_{u_i \in d^-_j} w'_{ij}.
\end{align}

Note that the above definitions, as given in Eq.~\eqref{eq:bias},
Eq.~\eqref{eq:tw} and Eq.~\eqref{eq:rating} depend on each other.  Later, we show how to solve such a system of interdependent equations using an
iterative algorithm.

\subsubsection{Connection with Co-reference and Co-citation Matrices}

We now analyze the bias and true rating formulation using matrix algebra, and see an intriguing connection
with the co-reference and co-citation matrices.
Notice that there is a non-linearity involved in 
the formulation of true rating in the form of \emph{max} and \emph{min} operators. For simplicity, we ignore these terms in the current section, but
later we given an unconditional proof of
convergence and uniqueness.

Suppose $\vec{b}$ and $\vec{r}$ are the bias and true rating vectors and $W$ is
the adjacency matrix.  Let $D^+=diag(\left | d^+_1 \right | ,\left | d^+_2
\right | ,\left | d^+_3 \right |, \ldots)$ and  $D^-=diag(\left | d^-_1 \right |
,\left | d^-_2 \right | ,\left | d^-_3 \right |, \ldots)$ be the two diagonal degree
matrices.
Further, let $C$ be the (symmetric) connection matrix with $0/1$
entries, i.e., if a user rates an item, there is a $1$ entry in the
matrix; otherwise, it is $0$.  We can now write the bias and true rating vectors as follows:
\begin{align}
	\vec{b}& =(D^+)^{-1}(W\vec{1}-C \vec{r}) \\
	\vec{r} &=(D^-)^{-1}(W^T\vec{1}-\alpha C^T \vec{b})
\end{align}
where $(D^*)^{-1}$ denotes the Moore-Penrose pseudoinverse of the
degree matrix.

The analytical solution for bias can be written as:
\begin{align}
	\vec{b}& =\left (I-\alpha \cdot (D^+)^{-1}C(D^-)^{-1}C^T  \right ) ^{-1}\vec{m}\\
		&=\left (I-\alpha \cdot C'C^{''T}  \right ) ^{-1}\vec{m}
\end{align}
Here, $C'$ and  $C''$ are two degree normalized connection matrices and 
$\vec{m}=(D^+)^{-1}(W-C(D^-)^{-1}W^T)\vec{1}$. For $(I-\alpha \cdot C'C^{''T}  )^{-1}$ to exist,
we need $\alpha<1$. This constraint on $\alpha$  would ensure that the spectral radius is less than $1$. This is the reason we
put a constraint on $\alpha$. 

Notice that $ C'C^{''T}$ is the co-reference matrix, thus our formulation bears a significant similarity with the Randomized HITS \cite{ng01}, and
matrix equations are of similar forms. We can similarly analyze true rating vector which relies on the co-citation matrix.

\subsection{Iterative algorithm}

The iterative algorithm proceeds by updating the bias and rating of nodes in
each iteration.  Given the bias of users obtained in the previous iteration
(say, iteration $t$), the true ratings of objects are computed for the current
iteration (i.e., iteration $t+1$).  Using these values in turn, the biases are
again computed for the current iteration.  Denoting the bias of user $u_i$ and
true rating of object $r_j$ at iteration $t$ by $b_i^t$ and $r_j^t$
respectively,
\begin{align}
	r_j^{t+1} &= \frac{1}{|d^-_j|} \sum_{u_i \in d^-_j} \max\big\{0, \min\{1, w_{ij} - \alpha \cdot  b^{t}_i\}\big\} \\
	b^{t+1}_i &= \frac{1}{|d^+_i|} \sum_{o_j \in d^+_i} \big( w_{ij} - r^{t+1}_j \big)
\end{align}

The above two equations can be combined to obtain a single recursive equation for bias:
\begin{align}
	b^{t+1}_i &= \frac{1}{|d^+_i|} \sum_{o_j \in d^+_i} \Big(w_{ij} - \frac{1}{|d^{-}_j|} \sum_{u_i \in d^-_j} \max\big\{0, \min\{1, w_{ij} - \alpha \cdot  b^{t}_i\}\Big) \nonumber
\end{align}

The algorithm proceeds till it reaches the stopping criterion.  It may stop when
the number of iterations reach a pre-defined threshold or when the change in
values of bias and rating fall below another pre-defined threshold.  In
Section~\ref{sec:convergence}, we show that the errors fall below a given
threshold after at most a certain number of iterations.

The initial bias and rating values are assigned randomly.  In
Section~\ref{sec:uniqueness}, we show that the initial values do not matter, and
the final values obtained after convergence of the algorithm are unique.

\subsection{Complexity}

We now analyze the runtime complexity of the above iterative algorithm.  In
each iteration, the values of bias and rating are updated for all nodes.
Updating the rating of an object $o_j$ requires $O(|d^-_j|)$ time.  For all
objects, the amortized total time taken is $O(\sum_{o_j} |d^-_j|) = O(m)$, where $m$ is
the number of edges in the graph.  Similarly, updating the bias of user $u_i$
requires $O(|d^+_i|)$ time.  The amortized total time taken for all users is $O(\sum_{u_i}
|d^+_i|)$ which is again $O(m)$.  Thus, the time spent per iteration is $O(m)$.
If the algorithm is run for $k$ iterations, the total running time is $O(km)$.

Although $m$ can be as large as $|U|\cdot |O|$, for most practical datasets, the
users rate only a handful of objects.  Hence, in practice, $m \ll |U|\cdot |O|$.  In
general, $m$ is more close to $n = O( \max\{|U|, |O|\})$.  We show in
Section~\ref{sec:convergence} that the number of iterations $k$ is logarithmic
in the (inverse) error term.  Thus, the algorithm is very fast.

\section{Properties of the algorithm}
\label{sec:prop}

In this section, we characterize certain useful properties of the algorithm
including error bounding of values and convergence to unique values. 
In this section, we show that error after any iteration is bounded. This helps
in fixing the maximum number of required iterations apriori.
	Before we proceed to the proof, we  present four facts. Namely\\
\\
	\textbf{Fact 1:} If $w-\alpha \cdot  b_i\geq 1 \implies$
	$\max(0,\min(1,w-\alpha \cdot  b_i))=1.$\\ 
	\textbf{Fact 2:} If $w-\alpha \cdot b_i\leq 0 \implies$
          $ \max(0,\min(1,w-\alpha \cdot  b_i))=0.$\\ 
	\textbf{Fact 3:} If  $0\leq(w-\alpha \cdot b_i)\leq 1 \implies$
	$\max(0,\min(1,w-\alpha \cdot  b_i))=(w-\alpha \cdot b_i)$.\\
	\textbf{Fact 4:} $|b_i^*-b_j^*|\leq 2$.

We now prove the following lemma that states that the difference between the
ratings are bounded by the difference between the different values of bias.  It
is used in later proofs.

\begin{lemma}
	Suppose $b^m_i$ and $b^n_i$ are two values of bias of node $i$. These can
	be the values in 
	two different iterations $m$ and $n$ or can also be
	two different biases when the initial conditions are different, i.e., we start
	with two different seed values of bias. We prove the following inequality:
	\begin{align*}
	M&=| \max(0,\min(1,w-\alpha \cdot  b^m_i)-\max(0,\min(1,w-\alpha \cdot b^n_i)|   \\
		&\leq \alpha \cdot |b^m_i-b^n_i|
	\end{align*}
\end{lemma}
\begin{proof} We prove this for six possible cases:
	\begin{enumerate}
		\item When $(w-\alpha \cdot  b^m_i)\geq 1$ and  $(w-\alpha \cdot  b^n_i)\geq
			1$: \\ Using Fact 1 we can observe that $M=0$. Since both $\alpha$
			and $|b^m_i-b^n_i|$ are positive  quantities, we can
			represent this in an inequality form as $M\leq
			\alpha\cdot |b^m_i-b^n_i|$.
		\item When $(w-\alpha \cdot  b^m_i)\geq 1$ and  $0 \leq (w-\alpha \cdot  b^n_i)
			\leq 1$: \\ Using Fact 1, we get $M=1-(w-\alpha \cdot  b^n_i)$. We know
			that $(w-\alpha \cdot  b^m_i)\geq 1$.  Therefore, $M\leq ((w-\alpha
			.  b^m_i) -(w-\alpha \cdot  b^n_i))$.  Hence $M\leq
			\alpha\cdot |b^m_i-b^n_i|$.
		\item When $(w-\alpha \cdot  b^m_i)\geq 1$ and  $(w-\alpha \cdot  b^n_i) \leq
			0$: \\ From Fact 1 and Fact 2, we get $M=1$.  Since $(w-\alpha 
		\cdot	b^m_i)\geq 1$ and $(w-\alpha \cdot  b^n_i) \leq 0$,
			$\alpha\cdot |b^m_i-b^n_i|\geq1$. Therefore
			$M\leq\alpha\cdot |b^m_i-b^n_i|$.
		\item When $0 \leq (w-\alpha \cdot  b^m_i) \leq 1$ and  $0 \leq (w-\alpha \cdot 
			b^n_i) \leq	1$: \\ Using Fact 3, we obtain $M=(w-\alpha \cdot  b^m_i)-(w-\alpha
			\cdot  b^n_i)= \alpha.|b^m_i-b^n_i|$. In other words, $M\leq
			\alpha\cdot |b^m_i-b^n_i|$.
		\item When $0\leq(w-\alpha \cdot  b^m_i)\leq1$ and  $(w-\alpha \cdot  b^n_i)\leq0$:
			\\ Using Fact 3 and Fact 2, we get $M= (w-\alpha \cdot  b^m_i)$.
			Now, let us consider a quantity $N=|(w-\alpha \cdot  b^m_i) -
			(w-\alpha \cdot  b^n_i)|$. Since, $(w-\alpha \cdot  b^n_i)\leq0$,
			therefore, $N=|(w-\alpha \cdot  b^m_i)| + |(w-\alpha \cdot 
			b^n_i)|$.  Notice that $N\geq M$ and we can rewrite $N$ as
			$\alpha \cdot |b^m_i-b^n_i|$. Therefore, $M\leq
			\alpha \cdot |b^m_i-b^n_i|$.
		\item When $(w-\alpha \cdot  b^m_i)\leq0$ and  $(w-\alpha \cdot  b^n_i)\leq0$: \\
			Using  Fact 2, we get $M=0$. Since  $\alpha\cdot |b^m_i-b^n_i|$
			is a positive quantity, therefore, $M\leq
			\alpha\cdot |b^m_i-b^n_i|$.
	\end{enumerate}
	\hfill{}
\end{proof}

\subsection{Error-bound}

In this section, we bound the error; specifically, we show that the difference
of bias of two consecutive iterations is bounded and it decreases exponentially as
we increase the number of iterations.

Suppose $b^t_i$ is the bias value of node $i$ after $t$ iterations. We prove the
following theorem on error that states that the difference is bounded by an
exponential factor of $\alpha$.

\begin{theorem}
	The difference of bias of a node at any iteration $t$ from the next
	iteration is bounded by an inverse exponential function of $t$:
	\begin{align}
	\label{eq:bound}
	|b^{t+1}_i - b^t_i| \leq {2 \cdot \alpha^{t+1}}.
	\end{align}
\end{theorem}
\begin{proof}
	We use mathematical induction for the proof.
	\textbf{Basis}: We first prove for $t=0$.
	\begin{align*}
	&|b^{1}_i  - b^{0}_i|\\
	&= \Big|\frac{1}{|d^{+}_i|} \sum_{j \in d^{+}_i}\Big (
		\frac{1}{|d^{-}_j|} \sum_{k \in d^{-}_j}
		\Big(\max(0,\min(1,w_{kj}- \alpha \cdot b^1_k)) \\
	  &\qquad \qquad  - \max(0,\min(1,w_{kj}- \alpha \cdot  b^0_k)) \Big) 
\Big )\Big| \\
	&\leq \Big |\frac{1}{|d^{+}_i|} \sum_{j \in d^{+}_i}\Big (
		\frac{1}{|d^{-}_j|} \sum_{k \in d^{-}_j}
		\Big|\max(0,\min(1,w_{kj}- \alpha \cdot b^1_k))\\
	&\qquad \qquad -\max(0,\min(1,w_{kj}- \alpha\cdot  b^0_k))
	\Big| \Big)	 \Big| \\		
	&\leq \left |\frac{1}{|d^{+}_i|} \sum_{j \in d^{+}_i}\left (
		\frac{1}{|d^{-}_j|} \sum_{k \in d^{-}_j}
		(\alpha\cdot |b^2_k-b^1_k|)
		 \right )\right|  \mbox{       [Using Lemma 1]}\\
	&\leq \left |\frac{1}{|d^{+}_i|} \sum_{j \in d^{+}_i}\left (
		\frac{1}{|d^{-}_j|} \sum_{k \in d^{-}_j}
		(\alpha\cdot 2)
		 \right )\right| \leq 2\cdot \alpha \qquad \mbox{[Using Fact 4]}
	\end{align*}
	\textbf{Induction step}: We assume the bound to be true for $b^t_i$ in
	the $t^\text{th}$ iteration, i.e., $b^t_i - b^{t-1}_i \leq 2\cdot  \alpha^t$.  In the $(t+1)^\text{th}$ iteration,
	\begin{align*}
	&|b^{t+1}_i  - b^{t}_i| \\
	&=\Big| \frac{1}{2|d^{+}_i|} \sum_{j \in d^{+}_i}\Big(\frac{1}{|d^{-}_j|} \sum_{k \in d^{-}_j}
		\max(0,\min(1,w_{kj}+ \alpha \cdot b^{t}_k))\\
	& \qquad \qquad  -\max(0,\min(1,w_{kj}+ \alpha\cdot  b^{t-1}_k)))
		 \Big )\Big| \\		
	&\leq \left |\frac{1}{|d^{+}_i|} \sum_{j \in d^{+}_i}\left (
		\frac{1}{|d^{-}_j|} \sum_{k \in d^{-}_j}
		(\alpha\cdot |b^{t}_k-b^{t-1}_k|)
		 \right )\right| \\
	&\leq \left |\frac{1}{|d^{+}_i|} \sum_{j \in d^{+}_i}\left (
		\frac{1}{|d^{-}_j|} \sum_{k \in d^{-}_j} (\alpha\cdot 2 \cdot (\alpha)^t)
		\right )\right| \quad \mbox{[Induction assumption]}\\
	& = 2\cdot \alpha^{t+1}
	\end{align*}
	Thus, the error is bounded by an inverse exponential function of $t$.
\end{proof}

\subsection{Proof of Convergence}
\label{sec:convergence}
Using the above error bounds, we can observe that difference of two consecutive
iterations decrease  as we increase the number of iterations.  For any
$\epsilon > 0$, we have $t=\log_{1/ \alpha}(\frac{2}{ \epsilon})$ such that for any two iterations $m,n>t$, we have:
\begin{align*}
|b^{m} - b^n| \leq \epsilon
\end{align*}
The above sequence is a Cauchy sequence and thus converges.
\subsection{Proof of Uniqueness}
\label{sec:uniqueness}

We next establish the proof for the uniqueness of the system.  We
will prove it by contradiction. Assume that there are at least two values of
bias that
satisfy the bias equation.  Suppose $b^1_i$ and $b^2_i$ be two converged
values of node $i$ and their difference is $\delta(i)=b^1_i-b^2_i$.  Suppose,
$M$ denotes the largest such difference, i.e.,
$M=\max_i\delta(i)$.  If we can show
that $M$ is zero, then uniqueness is proved. In other words, there exist a single
solution to the bias equation.

\begin{theorem}
	The bias of a node converges to a unique value.
\end{theorem}
\begin{proof}
	\begin{align*}
	&M = 
	\Big| \frac{1}{2|d^{+}_i|} \sum_{j \in d^{+}_i}\Big(\frac{1}{|d^{-}_j|} \sum_{k \in d^{-}_j}
		\max(0,\min(1,w_{kj}- \alpha \cdot b^{1}_k))\\
	&\qquad\qquad\qquad-\max(0,\min(1,w_{kj}- \alpha \cdot b^{2}_k))
		\Big )\Big| \\
	&\leq \Big |\frac{1}{|d^{+}_i|} \sum_{j \in d^{+}_i}\Big (
		\frac{1}{|d^{-}_j|} \sum_{k \in d^{-}_j}
		\Big|\max(0,\min(1,w_{kj}- \alpha \cdot b^{1}_k))\\
	&\qquad\qquad\qquad-\max(0,\min(1,w_{kj}- \alpha \cdot b^{2}_k))
		 \Big |\Big)\Big| \\		
	&\leq \left |\frac{1}{|d^{+}_i|} \sum_{j \in d^{+}_i}\left (
		\frac{1}{|d^{-}_j|} \sum_{k \in d^{-}_j}
		(\alpha\cdot |b^{1}_k-b^{2}_k|)
		 \right )\right|     \mbox{[Using Lemma 1]}\\
	&\leq \left |\frac{1}{|d^{+}_i|} \sum_{j \in d^{+}_i}\left (
		\frac{1}{|d^{-}_j|} \sum_{k \in d^{-}_j}
		(\alpha\cdot |b^{1}_p-b^{2}_p|)
		 \right )\right|   \mbox{[By definition ]}\\
	&= \left |\frac{1}{|d^{+}_i|} \sum_{j \in d^{+}_i}\left (
		\frac{1}{|d^{-}_j|} \sum_{k \in d^{-}_j}
		(\alpha\cdot M)
		 \right )\right| \\
	&= \alpha\cdot M
	\end{align*}
	Since $\alpha \leq 1$ and M is a positive quantity, inequality $M \leq
	\alpha\cdot M$ holds when $M=0$.  This confirms uniqueness.
\end{proof}

\begin{table}[t]
\begin{center}\small
\begin{tabular}{|c|c|c|c|}
\hline
\multirow{1}{*}{Bin} & {Number of ratings} & \multicolumn{2}{c|}{ Count of  movies } \\
\cline{3-4}
& & Dataset 1 & Dataset 2 \\
\hline
1 & 1 & 114 & 13\\
2 & 2-3 & 131 & 26 \\
3 & 4-7 & 142 & 40\\
4 & 8-15 & 204 & 76\\
5 & 16-31 & 311 & 109\\
6 & 32-63 & 463 & 216\\
7 & 64-127 & 508 & 258\\
8 & 128-255 &  633 & 363\\
9 & 256-511 & 598 &  366\\
10 & 512-1023 &  406 & 273\\
11 & $>$1023 &  196 & 122\\
\hline
& Total &  3706 & 1862\\
\hline
\end{tabular}
\caption{Grouping of movies based on the number of ratings received.}
\label{tab:bin}
\end{center}
\end{table}

\begin{figure}[t]
	\centering
	\includegraphics[width=0.85\columnwidth]{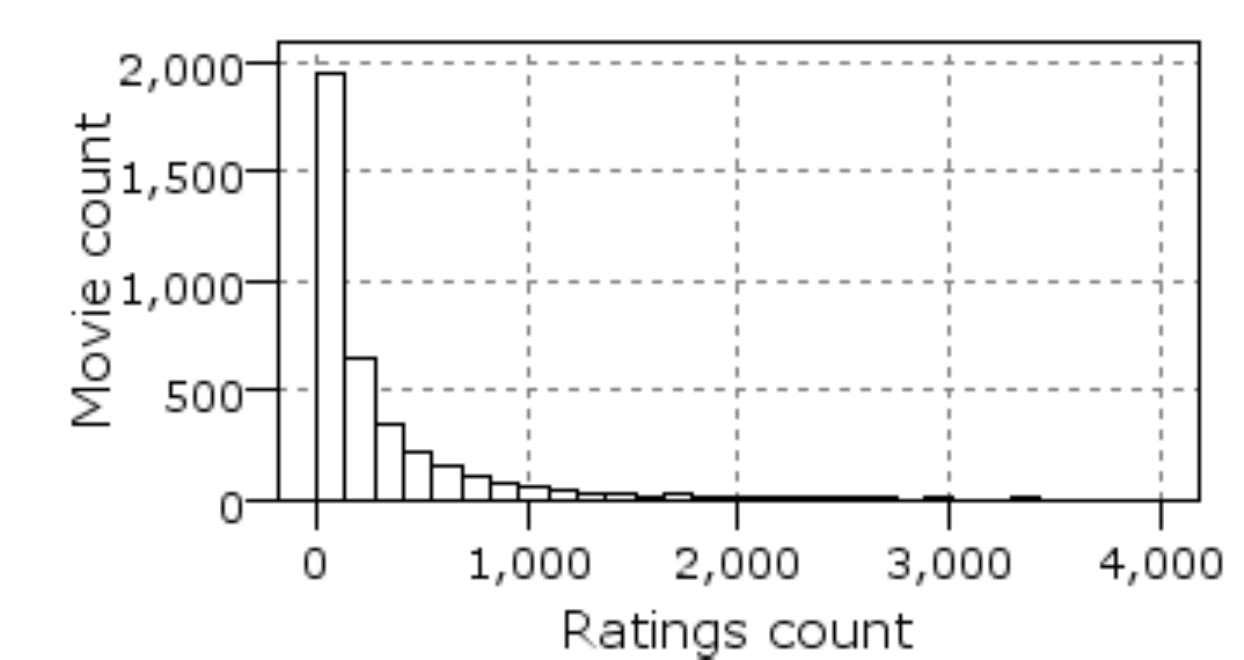}
	\caption{Distribution of number of movies with number of ratings received for Dataset 1. Most movies received very few ratings.}
	\label{fig:rating}
\end{figure}

\section{Experiments}
\label{sec:exp}

\subsection{Dataset description}

We use two datasets for our experiments. They are available from
\url{http://www.grouplens.org}.  The first dataset from
\url{www.grouplens.org/node/73} consists  of roughly 1 million movie ratings
given by around 6,000 users to around 3,700 movies.  The second dataset from
\url{http://www.grouplens.org/node/462} contains  extra information about
the movies such as ratings given by the top critics.  Thus we have both user
ratings and the net score given by the top critics.  We call the first dataset
as Dataset 1 and the second dataset as Dataset 2 throughout the paper.

In both the datasets, most of the movies receive only a few ratings.  For
example, in Figure~\ref{fig:rating}, we show the histogram of movies  over the
number of ratings they receive for Dataset 1.  We can notice that movies with
few ratings dominates the set. In Dataset 2, we observe similar behavior as it
is essentially a subset of Dataset 1 with extra information such as top critics
ratings. To gain more insights, we divide the two datasets into different bins
based on the number of ratings a movie receive.  This helps in  analyzing the
rating data in much greater detail. We divided all the movies in 11 groups (or
bins) based on the number of ratings they receive.  While binning, we
particularly focused on those movies which received fewer ratings.
Table~\ref{tab:bin} shows the distribution.  In both these datasets, we
normalize the movie ratings in the range $0$ to $1$, with $0$ implying the
lowest score and $1$ the highest.

\begin{figure}[t]
	\centering
	\includegraphics[width=0.7\columnwidth]{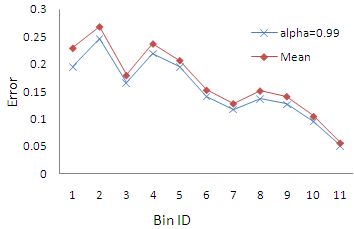}
	\caption{Mean Square Error for techniques on the Dataset 2.}
	\label{fig:errormean}
	\includegraphics[width=0.7\columnwidth]{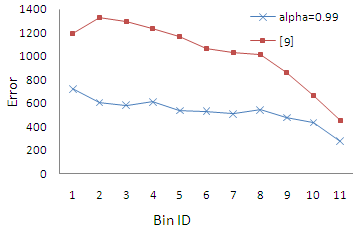}
	\caption{Absolute error in ranking computed over the Dataset 2.}
	\label{fig:errorrank}
\end{figure}
\subsection{Performance Evaluation}
\label{sec:eval}

In this section, we report the results against the baseline and existing techniques. Later (Section~\ref{sec:qual}) we analyze the 
behavior of our algorithm qualitatively.

We consider the top critics rating as the ground truth, and we compare
the various techniques against this metric. We compute the mean square error (MSE) and  error in ranking.
In  MSE, we simply look at the deviation in ratings returned by different techniques with respect to the
top critics rating. It has been shown that the critics score correlate well with the box office success of a movie in 
a long term, not so much with early box office receipts \cite{criticsboxoffice}. Therefore, it is fit to serve as a ground truth. 
To the best of our knowledge, Dataset 2 is the only publicly available dataset with critics ratings, and therefore in this set of experiments,
we use  this dataset alone.

 We show that our algorithm adheres well
with the critics' rating than the baseline method and a representative algorithm \cite{sdm07}. Mean rating is currently the most
widely used approach in most systems with small filtering mechanisms such as
frequent users. Therefore, it is a good candidate for  a baseline. The algorithm in \cite{sdm07}, similar to our algorithm,  exploits the global structure of
the graph rather than focusing locally.  The original algorithm suffers from the
divide-by-zero problem and the authors suggested using a small value; so, for
our comparison, we used $0.01$.

Note that the algorithm proposed  in~\cite{sdm07} does not converge and the
final values (which are out of scale) depend on the initial seeds.  For a fair comparison with this
algorithm, we consider another metric---error in ranking. Initially, we sort
the movies based on the top critics score to know their ranks and then we simply
look at the average absolute distance between the ranking returned by the
technique (our algorithm and \cite{sdm07}) and critics ranking.

The mean square error with mean rating on Dataset 2 is $0.142$ while our
algorithm ($\alpha=0.99$) reduces it to $0.129$.  In later sections, we analyze
the behavior of our algorithm for different values of $\alpha$. For now, we
assume that we are trying to remove as much bias as possible, and therefore, we
keep a high value for $\alpha$ (such as $0.99$). To gain more insight, we plot
the error across different bins Table~\ref{tab:bin} for these two techniques in
Figure \ref{fig:errormean}. We can observe that our algorithm gives superior
results across all the bins, especially in the ones that receive few ratings.
We can also observe that the error reduces as the number of ratings per bin
increases. This phenomenon is explained in detail in
Section~\ref{sec:numrating}.

Algorithm proposed in \cite{sdm07} does not converge, therefore we compare the
ranking returned by our algorithm against those as authors \cite{sdm07} originally suggested.  In Figure
\ref{fig:errorrank}, we can notice that our algorithm outperforms the 
algorithm. Also, the error in ranking decreases as the number of ratings
increases.

\begin{figure}[t]
\centering

\subfigure[Bias distribution for $\alpha=0.99$.]{
\includegraphics[scale=0.5]{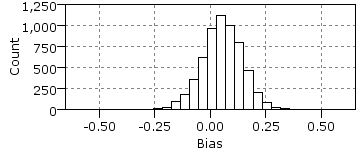}
\label{fig:subfig1}
}
\subfigure[Bias distribution for~\cite{sdm07}.]{
\includegraphics[scale=0.5]{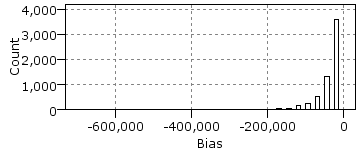}
\label{fig:bias1}
}
\caption{Distribution of bias.}
\label{fig:bias}
\end{figure}

\begin{figure}
\centering
\subfigure[True rating distribution for $\alpha=0.99$.]
{
\includegraphics[scale=0.5]{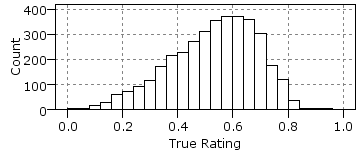}
\label{fig:subfig1}
}
\subfigure[True rating distribution for~\cite{sdm07}.]
{
\includegraphics[scale=0.5]{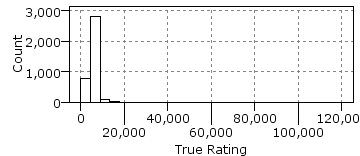}
\label{fig:rating1}
}
\caption{Distribution of true ratings.}
\label{fig:rating1}
\end{figure}

\subsection{Qualitative Evaluation}
\label{sec:qual}

In this section, we look at other properties of our algorithm, mean rating and
the algorithm
proposed in~\cite{sdm07} such as the distribution of bias, effect of number of
ratings a movie receives,
convergence analysis of the algorithm, and comparison of relative rankings. As Dataset 2 is simply a smaller subset of
Dataset 1 with the additional critics information, we now show results directly on the Dataset 1.

\begin{figure*}
\centering
\subfigure[Variation of ratings across bins.]
{
\includegraphics[scale=0.65]{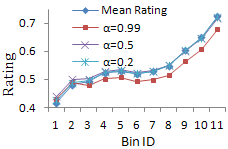}
\label{fig:variation}
}
\subfigure[Deviation of ratings across bins.]
{
\includegraphics[scale=0.65]{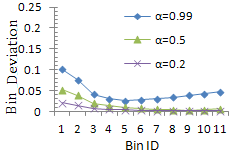}
\label{fig:deviation}
}
\subfigure[Relative deviation of ratings across bins.]
{
\includegraphics[scale=0.65]{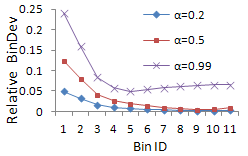}
\label{fig:relative}
}
\caption{Comparison of ratings across different bins.}
\label{fig:numrating}
\end{figure*}
\subsubsection{Distribution of bias and true rating}

In this experiment, we look at the distribution of bias and true rating obtained
by running our algorithm and algorithm in~\cite{sdm07}.
Figure~\ref{fig:bias} shows the distribution of bias at the end of $10$
iterations. Our algorithm gives a bell curve type distribution.
Note that the algorithm in~\cite{sdm07} did not converge.  We observed small changes in the bias distribution 
for different values of $\alpha$ such as $0.2$, $0.5$ and $0.99$.
The bias values obtained by the algorithm in~\cite{sdm07} are all negative and out of scale.  Our bias values, on the other hand,
quantify the effect by which the rankings etc. are shifted, and are therefore,
more close to actual biases.

Figure~\ref{fig:rating1} shows similar distribution of true ratings at the end of
$10$ iterations for our method and the algorithm in~\cite{sdm07}.  
Even with high $\alpha$ our algorithm gives meaningful value of true rating.
  Once again,
the actual values of true ratings obtained using the algorithm in~\cite{sdm07}
are useful only for relative ranking purposes and the absolute values do not
convey anything.  The ratings obtained by our method are more indicative of the
actual ratings.

\subsubsection{Effect of number of ratings}
\label{sec:numrating}
As discussed earlier, there are many movies for which the number of ratings is
quite less.  There are many reasons for fewer ratings.  A movie that has been
marked as ``poor'' by a user is likely to be watched by less number of people
and will thus have only a few ratings.  On the other hand, when a movie is rated
``excellent'', more people watch/buy it, and the movie receives more ratings.
This behavior is loosely termed as \emph{purchasing bias}.  Thus, we expect a movie
with a large number of reviews to have high ratings, and a movie which is rated
by few users to be lowly rated in general.  Experiments also reveal such
behavior.

Figure~\ref{fig:variation} shows the variation of the true ratings obtained
using different $\alpha$ values.  It also plots the mean of the original
ratings.  As predicted by the purchasing bias, movies with more ratings
generally have higher values.    We
did not plot the ratings obtained using the algorithm in~\cite{sdm07}, as they
are very high and out of scale and their absolute values do not signify
anything.  Instead, later we compare with the relative ranking as the authors
in~\cite{sdm07} intend to do.
 
The next graph, Figure~\ref{fig:deviation}, plots the deviations for each bin.
It captures the average difference of the true ratings from the mean rating.
The deviation of bin $k$ is
\begin{align*}
	bindev(k) = \frac{1}{|bin(k)|} \sum_{o_j \in bin(k)} |r_j - meanRating_j|
\end{align*}
where $|bin(k)|$ is the number of movies in bin $k$ and $meanRating_j$ is the
mean rating of movie $o_j$.

The deviation is high for movies that receive fewer ratings and is much low for
movies that are rated by many users.  This can be explained by the ``law of
large numbers'' that states that the average of a large number of values is
closer to the expected value.  The second important observation is the fact that
the deviation is higher for larger $\alpha$.  The parameter $\alpha$ denotes the
quantity that is removed from the mean rating.  When $\alpha = 0$, the user is
never biased, and the true rating is equal to the mean rating.  With a larger
amount of shift, the deviation is larger as well.  Finally,  it seems
that the deviation in later bins (those with more ratings) are higher.
 It is because average true rating also increases in the later bins. However, we do 
not expect  the \emph{relative} deviation to be as large.

To understand this behavior better, we next compute the relative deviation in
each bin as the ratio of deviation to the true rating:
\begin{align*}
	relbindev(k) = \frac{1}{|bin(k)|} \sum_{o_j \in bin(k)} \frac{|r_j - meanRating_j|}{r_j}
\end{align*}
Figure~\ref{fig:relative} shows the graph.

Indeed, the relative deviations in bins with fewer ratings are much higher.
While for movies having less than 4 ratings, the relative deviation can be as
high as 25\% , for those with more ratings, it stabilizes to 1\%.  Even
when $\alpha$ is very high ($0.99$), it is around 5\% only.

\begin{table}
\begin{center}\small
\begin{tabular}{|c||c||c|c|c||c|}
\hline
\multirow{2}{*}{Movie ID} & Mean & \multicolumn{3}{c||}{$\alpha = $} & \multirow{2}{*}{\cite{sdm07}} \\
\cline{3-5}
& Rating & $0.2$ & $0.5$ & $0.99$ & \\
\hline
\hline
3280&1 & 12 & 101 & 1417 & 4 \\
\hline
3233&2 & 3 & 3 &4 & 2021\\
\hline
1830&3 & 7 & 7& 9&19\\
\hline
3881&4 & 9 & 10 & 33&1641\\
\hline
3656&5 & 8 & 8 & 15& 2688\\
\hline
787&6 & 5 & 5 & 5& 1152\\
\hline
3607&7 & 6 & 6 & 10&10\\
\hline
3172&8 & 4 &  4 & 6&3456\\
\hline
3382&9 & 1 & 1 & 1& 1\\
\hline
989&10 & 2 &  2 & 3&3369\\
\hline
\end{tabular}
\caption{Top ranked movie and their IDs upon sorting by mean rating.}
\label{tab:rankmean}

\begin{tabular}{|c||c|c|c||c||c|}
\hline
\multirow{2}{*}{Movie ID}  & \multicolumn{3}{c||}{$\alpha = $} & \multirow{2}{*}{Mean} &\multirow{2}{*}{\cite{sdm07}} \\
\cline{2-4}
 & $0.99$ & $0.2$ & $0.5$ & Rating & \\
\hline
\hline
3282  &1  & 1 & 1 &  9 & 1 \\
\hline
557   &2  & 15& 9 &  23    & 5\\
\hline
989   &3  & 2 & 2 & 10    &3369\\
\hline
3233  &4  & 3 & 3 & 2   &2021\\
\hline
787   &5  & 5 & 5 &  6   & 1152\\
\hline
3172  &6  & 4 & 4 &  8    & 3456\\
\hline
578   &7  & 17& 14&  25   &9\\
\hline
2503  &8  & 13& 13 &13    &34\\
\hline
1830  &9  & 7 & 7 & 3    & 19\\
\hline
3607  &10 & 6 & 6 & 7    &10\\
\hline
\end{tabular}
	\caption{Top ranked movie and their IDs upon sorting by true rating for $\alpha=0.99$.}
	\label{tab:rank99}

\begin{tabular}{|c||c|c|c||c||c|}
\hline
\multirow{2}{*}{Movie ID}  & \multicolumn{3}{c||}{$\alpha = $} & \multirow{2}{*}{Mean} &\multirow{2}{*}{\cite{sdm07}} \\
\cline{2-4}
 & $0.2$ & $0.5$ & $0.99$ & Rating & \\
\hline
\hline
2858    & 79   & 91  &  111  & 76 & 209 \\
\hline
260     & 35   &  35 &   50 & 34    & 208\\
\hline
1196    & 98   & 108 & 149  & 91    &318\\
\hline
1210    & 377  & 407 &   486 & 345   &616\\
\hline
480     & 889  & 936 & 1016  & 861   & 1080\\
\hline
2028    & 67   & 77  & 107   & 64    & 316\\
\hline
589     & 320  & 354 &  417  & 300   &594\\
\hline
2571    & 82   & 96  & 126  &77    &205\\
\hline
1270    & 429  & 455 &  529  &438    & 838\\
\hline
593     & 60   & 62  & 76  & 60    &277\\
\hline
\end{tabular}
	\caption{Movie are sorted (descending) based on the number of ratings they received.}
	\label{tab:ranktopvotes}

\begin{tabular}{|c||c|c|c||c||c|}
\hline
\multirow{2}{*}{Movie ID}  & \multicolumn{3}{c||}{$\alpha = $} & \multirow{2}{*}{Mean} &\multirow{2}{*}{\cite{sdm07}} \\
\cline{2-4}
 & $0.2$ & $0.5$ & $0.99$ & Rating & \\
\hline
\hline
127    & 3674   & 3634  &  3225  & 3678 & 3686 \\
\hline
133     & 3670   &  3544  &  2355 & 3679    & 3696\\
\hline
139    & 450   & 529 & 709   & 374    &3555\\
\hline
142    & 3696  & 3691 &   3687   & 3684   &3688\\
\hline
226     & 3529  & 3536 &  3573  & 3520   & 3551\\
\hline
286    & 2369   & 2145  &  1727     & 2464    & 94\\
\hline
311     & 2358  & 2112 &  1652   & 2466   &3477\\
\hline
396    & 440   & 503  &  654    &376    &38\\
\hline
398    & 366  & 333 & 286    &377    & 3380\\
\hline
402     & 2400   & 2216  &   1857    & 2543    &37\\
\hline
\end{tabular}
	\caption{Movies which received only $1$ vote.}
	\label{tab:rankfewvotes}
\end{center}
\end{table}

\subsubsection{Comparison of relative ratings}

A ranked solution has also been proposed by the authors of~\cite{sdm07}. They have shown that
relative ranking could still be preserved even if the algorithm does not converge.
  In this section, we compare the ranking of
movies based on the ratings obtained using the algorithm in~\cite{sdm07} and our
algorithm.

The first table, Table~\ref{tab:rankmean}, ranks the movies according to the
mean rating they received originally.  When $\alpha$ is small, there is little
effect of bias, and our rankings closely resemble that of mean rating.  
In general, we do not expect a movie with a high mean rating to become a very 
bad movie after removing bias. There will be some difference and as $\alpha$
increases, such differences may increase.
Still, we do not expect a dramatic change. 
Even for very large $\alpha$ ($0.99$), the list obtained is better than that
using~\cite{sdm07}.  Many of the top-10 movies in that list are ranked in
thousands as opposed to only one such case for our algorithm.  In general, our
results are more stable.

In the next table, Table~\ref{tab:rank99}, we first sort the movies according to
the ratings obtained using $\alpha = 0.99$, and then tabulate their rankings
using the other schemes.  Once more it can be observed that all the ranked lists
are similar except for the algorithm in~\cite{sdm07} as it shows instabilities.

Movies that are rated by many users generally receive higher ratings.  In
Table~\ref{tab:ranktopvotes}, movies are sorted in a descending order according
to the number of ratings they have received.  Most largely rated movies are
indeed highly rated.  All the algorithms including the one from~\cite{sdm07}
rank them high (within 15\%, but mostly within 10\%).

It is, however, more interesting to analyze how the rankings behave for movies
that receive very few ratings.  The final table, Table~\ref{tab:rankfewvotes},
show the results for movies that have received only a single vote.  We have
shown only 15 movies sorted according to their IDs.  As indicated by the
purchasing bias, these movies are poorly rated, both according to mean rating
and our algorithm.  The algorithm in~\cite{sdm07}, shows some inconsistent
results and ranks certain poorly rated movies very highly (e.g., movie IDs 286,
396, and 402).  Our algorithms show more meaningful and stable ranking.

\begin{figure}
\centering
\subfigure[True Rating.]
{
\includegraphics[scale=0.50]{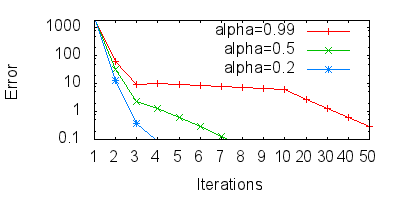}
\label{fig:errorrating}
}
\subfigure[Bias.]
{
\includegraphics[scale=0.50]{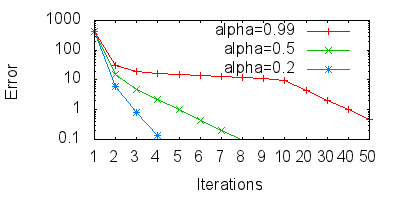}
\label{fig:errorbias}
}
\caption{Convergence plot: error versus number of iterations.}
\label{fig:error}
\end{figure}  
\subsubsection{Convergence}

In this set of experiments, we empirically measure the rate at which our
algorithm converges.  Figure~\ref{fig:error} shows the rate of convergence for
both true rating and bias.  The error is measured as the \emph{sum} of
differences in bias (or true rating) values over all nodes between two
consecutive iterations. The algorithm in~\cite{sdm07} did not converge in our setting.

We choose $0$ as the initial seed value of bias for all nodes.  In the first
iteration, this value is used to compute true ratings and in turn the new bias
values.  Since seed bias is $0$, bias computation is independent of $\alpha$ for
first iteration.  From second iteration onward, $\alpha$ starts to make a
difference in the convergence rate of bias and true rating.

In the graphs, we have plotted the error values up to 50 iterations.  We use the $L_1$ norm
to compute error.
Note that an error value, say $0.1$, as shown in the graph is the sum of error values over
3,000 nodes.  Consequently, the error for a single node is very  low.

The rate of convergence is slower for a larger $\alpha$.  When $\alpha \to 1$,
the error in each node can be as large, and therefore, the rate
of convergence will be  slow.  On the other hand, for a small $\alpha$, the
values converge very rapidly, and the algorithm is very  practical.

Since it has been shown that the error falls exponentially, the stopping
criterion can also be chosen to be when the difference in values between two
consecutive iterations becomes less than a threshold $\epsilon$.  The number of
iterations is a logarithmic function in the inverse of the threshold $\epsilon$.

\section{Conclusion}
\label{sec:conc}
In this work, we have proposed a novel technique to compute the bias of users
and consequently true ratings for products. People have different views when
they rate a product and it is necessary to capture their bias. By factoring
bias, we obtain true ratings of a product.  The bias and true rating values
computed by our algorithm are meaningful and can be associated directly with the
ratings user provide.  Our algorithm is iterative, fast and has several nice
properties including convergence to a unique solution and bounded errors.  The
maximum number of required iteration can also be fixed apriori depending on the
level of accuracy required. In experiments, we observe that our technique
produces consistent and good results. In future, we would like to analyze the 
change in user biases over time.  For example, users tend to give high ratings right after a movie release.

\bibliographystyle{abbrv}
\bibliography{sigir}




\end{document}